\DocumentMetadata{%
	lang=en-US,
	pdfstandard	= A-3u,
	pdfversion	= 1.7,
}

\documentclass[balance,nocopyright,upint,varvw,hyphenate,barcolor=Goldenrod3,german,nolists]{asmejour}%
\allowdisplaybreaks 

\hypersetup{%
	pdfauthor={John H. Lienhard},                       		   	
	pdftitle={ASME Journal Paper LaTeX Template},                  	
	pdfkeywords={ASME journal paper, LaTeX template, BibTeX style, asmejour class},
	pdfsubject = {Describes the asmejour LaTeX template},			
}

\JourName{ASME Letters in Dynamic Systems}

\newtheorem{theorem}{Theorem}
\newtheorem{problem}{Problem}
\newtheorem{assumption}{Assumption}
\newtheorem{definition}{Definition}

\usepackage{eso-pic}

\newcommand{\AuthorPreprintNotice}{%
  \AddToShipoutPictureFG*{%
    \AtPageLowerLeft{%
      \raisebox{10mm}{%
        \hspace*{15mm}%
        \parbox{180mm}{%
          \footnotesize
          \textit{Author preprint notice: This manuscript is the author-prepared preprint version originally submitted to the MECC--ALDSC joint review process. It has not been revised after the editorial decision, copyedited, typeset, assigned final publication metadata, or published by ASME. The final version of record, if published, will be available through the ASME Digital Collection.}%
        }%
      }%
    }%
  }%
}

\begin{document}
\AuthorPreprintNotice




\SetAuthorBlock{Yichao Wang$^\dagger$}{Department of Electrical and Computer Engineering,\\
	University of Connecticut, Storrs, CT 06269, USA.\\
	Email: yichao.wang@uconn.edu.
}

\SetAuthorBlock{Sameeha Tasneem$^\dagger$}{Department of Electrical and Computer Engineering,\\
	University of Connecticut, Storrs, CT 06269, USA.\\
	Email: lae25003@uconn.edu.
}

\SetAuthorBlock{Mohamadamin Rajabinezhad}{Department of Electrical and Computer Engineering,\\
	University of Connecticut, Storrs, CT 06269, USA.\\
	Email: mohamadamin.rajabinezhad@uconn.edu.
}

\SetAuthorBlock{Jinfeng Chen}{Department of Engineering Technology, University of Houston, Houston, TX 77004, USA.
   Email: particlefilter2012@gmail.com.
}

\SetAuthorBlock{Qin Lin}{Department of Engineering Technology, University of Houston, Houston, TX 77004, USA.
   Email: qlin21@central.uh.edu.
} 

\SetAuthorBlock{Kai Wang}{Connecticut Transportation Safety Research Center and Connecticut Transportation Institute, University of Connecticut, Storrs, CT 06269, USA.
   Email: kai.wang@uconn.edu
}

\SetAuthorBlock{Shan Zuo\CorrespondingAuthor}{Department of Electrical and Computer Engineering,\\
	University of Connecticut, Storrs, CT 06269, USA.\\
	Email: shan.zuo@uconn.edu.
}

\title{Resilient Control Lyapunov Function-based Quadratic Program for Quadrotors Under Cyberattacks}


\begin{abstract}
Ensuring the operational safety of quadrotors under partial actuator failures, lumped external disturbances, and malicious cyberattacks is a critical challenge due to the system's underactuated and highly nonlinear nature. Building on the existing result of a fault-tolerant control approach for a quadrotor experiencing a complete loss of two opposing rotors \cite{chen2024quadrotor}, this letter further addresses the additional challenge of malicious cyberattacks, which could be unknown and unbounded. While the baseline control law, rooted in proportional-derivative (PD) feedback and observer-based decoupling, effectively handles mismatched disturbances, it remains vulnerable to maliciously injected cyberattacks on the pseudo-control channels. To address this, a Resilient Control Lyapunov Function-based Quadratic Program (RCLF-QP) is developed, where a resilient compensational term with real-time online adaptation is designed in the conventional CLF to compensate for the maliciously injected unknown and unbounded attacks. Compared with the PD feedback control, the proposed QP-based constrained optimization control framework provides a systematic and
extensible framework that allows new control objectives and
constraints to be seamlessly integrated without altering the
underlying stability guarantees. The overall proposed controller integrates a model-based extended state observer with the proposed RCLF-QP mechanism to mitigate both lumped disturbances caused by aerodynamics and strong wind, and adversarial cyberattacks injected by malicious adversaries. Simulations in a high-fidelity environment demonstrate that the proposed RCLF-QP control architecture prevents trajectory divergence and system instability in scenarios where the baseline controller fails in maintaining the stability of Quadrotors under malicious attacks.
\end{abstract}

\date{}

\maketitle 
\revfootnote{$^\dagger$These authors contributed equally to this work.}


\section{Introduction} 
Operating in uncertain and adverser environments, the operational safety of quadrotors is increasingly threatened by the convergence of physical actuator failures, disturbances caused by aerodynamics and strong wind, and cyberattacks injected by malicious adversaries\cite{wu2025attack}. While traditional fault-tolerant control (FTC) effectively manages rotor loss, current research is shifting toward cyber-resilience to ensure continuity in the presence of sophisticated threats like false data injection attacks (FDIA) \cite{zhang2025distributed}. Significant progress has been made in handling rotor failures through methods like incremental nonlinear dynamic inversion (INDI) \cite{sun2021incremental}, robust adaptive sliding mode control \cite{wu2025robust}, and reinforcement learning-based approaches \cite{liu2024reinforcement}. A foundational advancement was established by Sun et al. \cite{sun2021incremental}, which proved the feasibility of maintaining flight despite the complete loss of two opposing rotors. Building upon this, Chen et al. \cite{chen2024quadrotor} introduced a model-based extended state observer (MB-ESO) for mismatched disturbance rejection, which utilizes system structural information to neutralize disturbances that do not enter through the control channels.

Despite these advances, most current methodologies, including the aforementioned baseline MB-ESO framework, rely on linear or proportional-derivative (PD) feedback loops for inner-loop tracking \cite{chen2024quadrotor}. While these approaches excel in handling environmental noise or bounded disturbances, they exhibit significant vulnerabilities when encountering adversarial FDIA. Specifically, conventional estimation-based defenses are often designed to isolate and compensate for disturbances within the feedback channel. However, when FDIA are injected directly into the pseudo-control input, such estimation methods may fail to provide the necessary regulation because the targets of their estimation are distinct from the point of attack. Recent state-of-the-art research has explored game-theoretical distributed formations \cite{geng2025resilient} and fixed-time state observers \cite{ranjan2025fixed} to enhance robustness. In \cite{shan2025co}, an observer-based gain-scheduling tracking control scheme is introduced for active suspension systems under malicious attacks. The work utilizes interval type-2 fuzzy rules and a co-design method to ensure $H_\infty$ performance and exponential stability in the presence of uncertainties and network-level threats. In addition to these methods, Control Barrier Functions (CBFs) have emerged as a powerful tool for enforcing safety constraints in the presence of actuator faults. For instance, in \cite{li2025safety}, a safety-critical control strategy for quadrotor UAVs is studied considering actuator faults and output constraints. The authors introduce a barrier Lyapunov function-based backstepping framework and a composite estimation module to jointly compensate for external disturbances and actuator malfunctions. In \cite{agrawal2022safe}, the synthesis of safe and robust observer-controller interconnections is studied for nonlinear systems with partial state information. The authors introduce novel CBFs based on input-to-state stable and bounded error observers to establish quadratic program-based controllers that certify safety despite measurement noise and bounded disturbances. In \cite{ye2026safety}, a safety-critical Model Predictive Control framework is studied for quadrotor UAVs operating in confined environments subject to compound disturbances and measurement errors. The authors introduce a refined disturbance observer to estimate wind and blade-induced airflow, alongside a measurement-robust tunable CBF that utilizes error upper bounds to improve system safety margins. In \cite{song2024safety}, the safety-critical fixed-time formation control of quadrotor UAVs is studied under disturbances and obstacle collision risks. The authors introduce a fixed-time distributed observer and a sliding mode-based disturbance observer, integrated with robust exponential CBFs, to ensure formation tracking and obstacle avoidance for systems with high relative degree. However, many of these frameworks assume that disturbances remain within specific bounds, leaving the system susceptible to catastrophic failure when unbounded FDIA lead to control command saturation or exceed the observer's bandwidth \cite{wu2025attack}.

The primary significance of this work lies in the development of a resilient control framework that enhances traditional tracking through a constrained optimization approach. This architecture provides two distinct advantages. First, by formulating the tracking objective as a Quadratic Program (QP) problem, we establish a versatile and flexible control structure. This optimization-based framework allows for the seamless integration of auxiliary control objectives. Second, we integrate a resilient compensational term directly into the RCLF-QP formulation with online real-time adapation. This term is specifically designed to dynamically neutralize the adverse effects of malicious attack injections. By treating the adversarial influence as an active perturbation within the stability constraints, the controller adaptively generates a compensational term that maintains the system's stability margin. This ensures that the quadrotor preserves trajectory integrity even under cyberattacks where state-of-the-art INDI and MB-ESO methods fail achieve the control objective.


\section{Preliminaries on Baseline Control Framework and Attack-Resilient Problem Formulation}\label{sec:model}

\textit{Notation.} Throughout this paper, boldface symbols such as $\boldsymbol{P}$ denote vectors or matrices, whereas non-boldface symbols like $p$ represent scalars. The superscripts $[\cdot]^I$ and $[\cdot]^B$ are used to indicate coordinates expressed in the inertial and body frames, respectively. Subscripts are employed to specify particular relationships or attributes; for instance, $\boldsymbol{x}_I$ represents the $x$-axis within the inertial frame. Additionally, $0_{m\times n}$ denotes an $m \times n$ zero matrix, and $\boldsymbol{I}_{m\times m}$ denotes an identity matrix of dimension $m$.

\vspace{-3mm}

\subsection{Quadrotor Dynamic Model}
Let $\mathcal{F}_\mathcal{I}=\{O_I, \boldsymbol{x}_I, \boldsymbol{y}_I, \boldsymbol{z}_I\}$ denote a right-hand inertial frame fixed to the ground, where $\boldsymbol{x}_I$, $\boldsymbol{y}_I$, and $\boldsymbol{z}_I$ point north, east, and downward, respectively \cite{sun2021incremental}. As illustrated in Fig. \ref{fig: bodyframe}, the right-hand body frame $\mathcal{F}_\mathcal{B}=\{O_B, \boldsymbol{x}_B, \boldsymbol{y}_B, \boldsymbol{z}_B\}$ is fixed to the quadrotor with its origin $O_B$ located at the center of mass; the axes $\boldsymbol{x}_B$, $\boldsymbol{y}_B$, and $\boldsymbol{z}_B$ point forward, right, and downward, respectively.

\begin{figure}[htbp]
\centering
\includegraphics[width=1.8in]{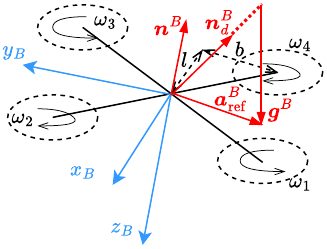} 
\caption{Visualizations of quadrotor and reduced attitude\cite{chen2024quadrotor}.}
\label{fig: bodyframe}
\end{figure}

The full nonlinear dynamical model is defined as follows 
\begin{subequations}\label{eq: quadrotor_nonlinear_dynamics}
    \begin{align}
    \dot{\boldsymbol{P}^I} &= \boldsymbol{V}^I \label{eq: translation_kinematics}\\
    m_v \dot{\boldsymbol{V}^I} &= m_v \boldsymbol{g}^I + \boldsymbol{R}_{IB}\boldsymbol{F}^B \label{eq: translation_dynamics}\\
    \dot{\boldsymbol{R}}_{IB} &= \boldsymbol{R}_{IB}\boldsymbol{\Omega}_{\times}^B \label{eq: rotation_kinematics}\\
    \boldsymbol{I}_{\boldsymbol{v}} \dot{\boldsymbol{\Omega}}^B &= -\boldsymbol{\Omega}_{\times}^B \boldsymbol{I}_{\boldsymbol{v}}\boldsymbol{\Omega}^B + \boldsymbol{M}^B \label{eq: rotation_dynamics}
    \end{align}
\end{subequations}
where \eqref{eq: translation_kinematics} and \eqref{eq: translation_dynamics} describe the translational dynamics, while \eqref{eq: rotation_kinematics} and \eqref{eq: rotation_dynamics} characterize the rotational dynamics. Here, $\boldsymbol{P}^I=[X, Y, Z]^T$, $\boldsymbol{V}^I=[V_x, V_y, V_z]^T$, and $\boldsymbol{g}^I=[0, 0, g]^T$ are the position, velocity, and gravity vectors in $\mathcal{F}_\mathcal{I}$ \cite{sun2021incremental}. The angular velocity vector in $\mathcal{F}_{\mathcal{B}}$ is defined as $\boldsymbol{\Omega}^B=[p, q, r]^T$, which includes the roll, pitch, and yaw rates. The rotation matrix from $\mathcal{F}_{\mathcal{B}}$ to $\mathcal{F}_{\mathcal{I}}$ is denoted by $\boldsymbol{R}_{IB}$, and $\boldsymbol{\Omega}^B_{\times}$ is the skew-symmetric matrix such that $\boldsymbol{\Omega}^B_{\times}\boldsymbol{a}=\boldsymbol{\Omega}^B\times \boldsymbol{a}$ for any vector $\boldsymbol{a}\in \mathbb{R}^3$ \cite{sun2021incremental}. The quadrotor's gross mass and inertia matrix are represented by $m_v$ and $\boldsymbol{I}_{\boldsymbol{v}}$, respectively.

The resultant force $\boldsymbol{F}^B$ and moment $\boldsymbol{M}^B$ acting on the center of mass in $\mathcal{F}_{\mathcal{B}}$ are adopted from \cite{sun2021incremental}.
The terms $\boldsymbol{F}_a$ and $\boldsymbol{M}_a$ represent lumped disturbances, which encompass model uncertainties, unmodeled aerodynamics, and adversarial signal injections \cite{sun2021incremental}.

\subsection{Reduced Attitude Control}
Under rotor failure conditions, the yaw angle becomes uncontrollable. Consequently, the reduced attitude control strategy \cite{sun2021incremental, mueller2016relaxed} simplifies the standard attitude control problem to a thrust vector pointing task, intentionally sacrificing full yaw regulation. As shown in Fig. \ref{fig: bodyframe}, $\boldsymbol{n}^B$ is a unit vector fixed to the body frame that serves as the rotation axis for the compromised quadrotor and aligns with the direction of the average thrust. To maintain position control, $\boldsymbol{n}^B$ must be aligned with the desired unit vector $\boldsymbol{n}_d^B$. This reference vector is initially generated as $\boldsymbol{n}_d^I$ in the inertial frame $\mathcal{F}_\mathcal{I}$ by the outer-loop controller and then transformed into the body frame $\mathcal{F}_\mathcal{B}$ \cite{sun2021incremental}. 

Since $\boldsymbol{n}_d^B$ is a unit vector with only two independent degrees of freedom, the standard rotation kinematics in \eqref{eq: rotation_kinematics} are replaced by the first two rows of the following relaxed attitude kinematic equations \cite{sun2021incremental}:
\begin{equation}\label{eq: reduced_attitude_kinematics}
    \begin{bmatrix}
        \dot{h}_1\\
        \dot{h}_2\\
        \dot{h}_3
    \end{bmatrix}=\begin{bmatrix}
        0 & r & -q\\
        -r & 0 & p\\
        q & -p & 0
    \end{bmatrix}\begin{bmatrix}
        h_1\\
        h_2\\
        h_3
    \end{bmatrix}+\begin{bmatrix}
        \lambda_1\\
        \lambda_2\\
        \lambda_3
    \end{bmatrix}
\end{equation}
where $\boldsymbol{n}_d^B=[h_1, h_2, h_3]^T$ and $\boldsymbol{R}_{IB}^T\dot{\boldsymbol{n}}_d^I=[\lambda_1, \lambda_2, \lambda_3]^T$ represents a disturbance vector, with $\boldsymbol{n}_d^I$ being the representation of $\boldsymbol{n}_d^B$ in $\mathcal{F}_\mathcal{I}$. For a quadrotor experiencing the complete loss of two opposing rotors, $\boldsymbol{n}^B=[0, 0, -1]^T$ is typically selected to maximize energy efficiency. To successfully align $\boldsymbol{n}^B$ with $\boldsymbol{n}_d^B$, both $h_1$ and $h_2$ must be stabilized to zero.

\subsection{Baseline Control Framework}\label{sec: proposed_framework}
To provide a clear analytical foundation for the design of the MB-ESO based inner-loop controller, the quadrotor's dynamics are represented as a control-affine nonlinear system influenced by a disturbance vector $\boldsymbol{w}$ \cite{sun2021incremental}:
\begin{equation}\label{eq: nonlinear_system_dynamics}
    \begin{cases}
        \dot{\boldsymbol{x}} = \boldsymbol{f}(\boldsymbol{x}) + \boldsymbol{G}(\boldsymbol{x})\boldsymbol{u} + \boldsymbol{Q}(\boldsymbol{x})\boldsymbol{w} \\
        \boldsymbol{y} = \boldsymbol{H}(\boldsymbol{x})
    \end{cases}
\end{equation}
In this model, the state vector is defined as $\boldsymbol{x}=[Z, V_z, h_1, h_2, p, q, r]^T$, while the control input vector is $\boldsymbol{u}=[\omega_1^2, \omega_2^2, \omega_3^2, \omega_4^2]^T$. The associated vector fields and matrices $\boldsymbol{f}(\boldsymbol{x}),\boldsymbol{G}(\boldsymbol{x}),\boldsymbol{Q}(\boldsymbol{x}),\:\text{and}\:\boldsymbol{H}(\boldsymbol{x})$ characterizing the system are adopted from \cite{sun2021incremental}.
The vector $\boldsymbol{w}=[F_{a,z}, \lambda_1, \lambda_2, M_{a,x}, M_{a,y}, M_{a,z}]^T$ serves as a generalized disturbance term that accounts for both environmental factors and potential adversarial interference. Here, $\boldsymbol{I}_{\boldsymbol{v}}=\text{diag}(I_x, I_y, I_z)$ denotes the diagonal inertia matrix, and $R_{33}$ represents the element in the third row and third column of the rotation matrix $\boldsymbol{R}_{IB}$. The model assumes that the gyroscopic moments produced by the rotors are negligible.

It is assumed that the system's full state is either directly measurable or can be reliably estimated. Since the quadrotor operates with only two remaining functional rotors, it is limited to the control of a maximum of two independent outputs. 
For the purposes of this investigation, we focus on the specific configuration where only rotors $2$ and $4$ are operational, implying $\omega_1^2 = 0$ and $\omega_3^2 = 0$.

The primary objective of the baseline control (MB-ESO) methods \cite{chen2024quadrotor} is to explore how to fully leverage available sensor data and structural model information to improve the estimation accuracy of the disturbance $\boldsymbol{w}$. Such high-fidelity estimation is essential for enhancing the performance and resilience of the two-input, two-output nonlinear system \eqref{eq: nonlinear_system_dynamics} when it is subjected to both environmental disturbances and malicious signal injections.
The baseline control architecture consists of an outer-loop horizontal position controller and an inner-loop altitude and attitude controller designed for resilience.


\subsubsection{Controller Design for Outer-Loop Horizontal Position}
To ensure a consistent and objective comparison with the INDI methodology \cite{sun2021incremental}, our framework utilizes an identical outer-loop PID control structure. This controller is responsible for generating the desired rotation axis $\boldsymbol{n}_d^I$ within the inertial frame $\mathcal{F}_{\mathcal{I}}$, which is determined as:
\begin{equation}\label{eq: desired_thrust_vector}
    \boldsymbol{n}_d^I=\dfrac{\boldsymbol{a}_{\text{ref}}^I-\boldsymbol{g}^I}{||\boldsymbol{a}_{\text{ref}}^I-\boldsymbol{g}^I||}
\end{equation}
where the reference acceleration vector $\boldsymbol{a}_{\text{ref}}^I$ is defined as:
\begin{equation}\nonumber\label{eq: reference_acceleration_definition}
    \boldsymbol{a}_{\text{ref}}^I=\begin{bmatrix}
        -k_p e_x-k_d \dot{e}_x -k_i \int e_x dt\\
        -k_p e_y-k_d \dot{e}_y -k_i \int e_y dt\\
        \ddot{Z}_{\text{ref}}
    \end{bmatrix}.
\end{equation}
In these expressions, $e_x=X-X_{\text{ref}}$ and $e_y=Y-Y_{\text{ref}}$ denote the horizontal position errors in the inertial frame $\mathcal{F}_{\mathcal{I}}$. The terms $k_p$, $k_i$, and $k_d$ represent the tunable PID gains, while $X_{\text{ref}}$, $Y_{\text{ref}}$, and $Z_{\text{ref}}$ correspond to the coordinates of the target reference trajectory.

\subsubsection{Controller Design for Inner-Loop Altitude and Attitude}
The inner-loop control architecture is developed utilizing an MB-ESO framework combined with disturbance decoupling to regulate the quadrotor's altitude and reduced attitude. This design ensures stability even in the presence of two opposing rotor failures and significant lumped disturbances.

\textit{Design of MB-ESO:} The primary objective of the MB-ESO is to maximize the utilization of sensor measurements and structural model knowledge to achieve high-fidelity estimation ($\boldsymbol{\hat w}$) of the disturbances ($\boldsymbol{w}$) within the dynamic model. Following the theoretical foundation in \cite{chen2023a}, the necessary and sufficient conditions for the existence of an MB-ESO are: 1) the system must be observable, and 2) there must be no invariant zeros located between the output and the disturbance. 

For the system defined in \eqref{eq: nonlinear_system_dynamics} with disturbance vector $\boldsymbol{w}$, the observer output vector is selected as $\boldsymbol{y}_o = [Z, h_1, h_2, p, q, r]^T$. It is verifiable that system \eqref{eq: nonlinear_system_dynamics} with output $\boldsymbol{y}_o$ can be transformed into a normal form characterized by a disturbance relative degree of $\{2,1,1,1,1,1\}$. The total degree sums to 7, which corresponds to the number of system states. Consequently, the system is observable and lacks zero dynamics between $\boldsymbol{w}$ and $\boldsymbol{y}_o$ \cite{isidori1995nonlinear}.


By defining each element of the disturbance vector $\boldsymbol{w}$ as an auxiliary state, the details of the six specialized MB-ESOs can be found in \cite{chen2024quadrotor}.

\textit{Control Design for Mismatched Disturbances:} To ensure the quadrotor follows the target trajectory $[Z_{\text{ref}}, h_{1, \text{ref}}\cos{\chi}+h_{2,\text{ref}}\sin{\chi}]^T$, the influence of the disturbance $\boldsymbol{w}$ on system \eqref{eq: nonlinear_system_dynamics} must be mitigated. However, because $\boldsymbol{w}$ enters the system through a different channel than the control input $\boldsymbol{u}$, it is classified as a mismatched disturbance and cannot be canceled directly. This study adopts the disturbance decoupling methodology to reject these mismatched terms \cite{yang2012nonlinear}.

System \eqref{eq: nonlinear_system_dynamics} is transformed into a normal form through the coordinate transformation $\boldsymbol{\Phi}(\boldsymbol{x}) = [\boldsymbol{\eta}, \boldsymbol{\xi}]^T$ \cite{isidori1995nonlinear}, where $\boldsymbol{\eta}$ represents the zero dynamics states \cite{sun2021incremental}, and:
\begin{equation}\label{eq: coordinate_transformation_xi}
\begin{gathered}
\boldsymbol{\xi} = [\xi_1, \xi_2, \xi_3, \xi_4]^T \\
= [H_1(\boldsymbol{x}), L_{\boldsymbol{f}}H_1(\boldsymbol{x}), H_2(\boldsymbol{x}), L_{\boldsymbol{f}}H_2(\boldsymbol{x})]^T \hfill
\end{gathered}
\end{equation}
In this expression, $H_i(\boldsymbol{x})$ is the $i$-th component of the output mapping $\boldsymbol{H}(\boldsymbol{x})$, and $L_{\boldsymbol{f}}H_i(\boldsymbol{x})$ denotes the first-order Lie derivative of $H_i(\boldsymbol{x})$ along the vector field $\boldsymbol{f}(\boldsymbol{x})$. Given that the zero dynamics are stabilizable through the proper selection of $\chi$ and that the control input $\boldsymbol{u}$ does not shift the transmission zeros, the control law is designed using the following subsystem \cite{yang2012nonlinear}:
\begin{equation}\label{eq: normal_form_subsystem}
\begin{gathered}
\dot{\boldsymbol{\xi}} = \boldsymbol{A}_c \boldsymbol{\xi} + \boldsymbol{B}_c[\boldsymbol{\alpha}(\boldsymbol{x}) + \boldsymbol{\mathcal{B}}(\boldsymbol{x})u] + \boldsymbol{D}_c(\boldsymbol{x})\boldsymbol{w} \hfill \\
\boldsymbol{y} = \boldsymbol{C}_c \boldsymbol{\xi} \hfill
\end{gathered}
\end{equation}
where $u = [\omega_2^2, \omega_4^2]^T$ represents the functional rotor inputs, and the system matrices are defined as:
\begin{equation*}
    \boldsymbol{A}_c=\begin{bmatrix} 0 & 1 & 0 & 0 \\ 0 & 0 & 0 & 0 \\ 0 & 0 & 0 & 1 \\ 0 & 0 & 0 & 0 \end{bmatrix}, \boldsymbol{B}_c=\begin{bmatrix} 0 & 0 \\ 1 & 0 \\ 0 & 0 \\ 0 & 1 \end{bmatrix}, \boldsymbol{C}_c=\begin{bmatrix} 1 & 0 \\ 0 & 0 \\ 0 & 1 \\ 0 & 0 \end{bmatrix}^T
\end{equation*}
The term $\boldsymbol{D}_c(\boldsymbol{x})$ explicitly characterizes the impact of disturbances on the output channels through Lie derivatives along the vector fields $Q_1, \dots, Q_6$:
\begin{equation*}
    \boldsymbol{D}_c(\boldsymbol{x})=\begin{bmatrix} 
        L_{Q_1}H_1 & L_{Q_2}H_1 & \cdots & L_{Q_6}H_1 \\ 
        L_{Q_1}L_{\boldsymbol{f}}H_1 & L_{Q_2}L_{\boldsymbol{f}}H_1 & \cdots & L_{Q_6}L_{\boldsymbol{f}}H_1 \\ 
        L_{Q_1}H_2 & L_{Q_2}H_2 & \cdots & L_{Q_6}H_2 \\ 
        L_{Q_1}L_{\boldsymbol{f}}H_2 & L_{Q_2}L_{\boldsymbol{f}}H_2 & \cdots & L_{Q_6}L_{\boldsymbol{f}}H_2 
    \end{bmatrix}
\end{equation*}
The remaining control-related terms are formulated as:
\begin{equation*}
\begin{gathered}
    \boldsymbol{\alpha}(\boldsymbol{x})=\begin{bmatrix} L_{\boldsymbol{f}}^2 H_1 \\ L_{\boldsymbol{f}}^2 H_2 \end{bmatrix}, \quad \boldsymbol{\mathcal{B}}(\boldsymbol{x})=\begin{bmatrix} L_{G_2}L_{\boldsymbol{f}}H_1 & L_{G_4}L_{\boldsymbol{f}}H_1 \\ L_{G_2}L_{\boldsymbol{f}}H_2 & L_{G_4}L_{\boldsymbol{f}}H_2 \end{bmatrix} \hfill
\end{gathered}
\end{equation*}
To neutralize the influence of the disturbance term $\boldsymbol{D}_c(\boldsymbol{x})\boldsymbol{w}$, the decoupling control law is defined as follows \cite{yang2012nonlinear}:
\begin{equation}\label{eq: decoupling_control_law}
\begin{gathered}
u = \boldsymbol{\mathcal{B}}^{-1}(\boldsymbol{x})[-\boldsymbol{\alpha}(\boldsymbol{x}) + u_0 + \boldsymbol{\Gamma}(\boldsymbol{x})\boldsymbol{\hat w}] \hfill
\end{gathered}
\end{equation}
where $u_0 = [-c_1^1 (\xi_1-Z_{\text{ref}}) -c_2^1 \xi_2, -c_1^2 \xi_3 -c_2^2 \xi_4]^T$ represents the baseline pseudo control component, which is designed based on state-feedback PD control. $\boldsymbol{\hat w}$ is the estimation of $\boldsymbol{w}$ using the MB-ESO in \cite{chen2024quadrotor}. The decoupling matrix $\boldsymbol{\Gamma}(\boldsymbol{x})$ is explicitly formulated to suppress the direct impact of $\boldsymbol{w}$ on the output derivative channels:
\begin{equation*}
\begin{gathered}
\boldsymbol{\Gamma}(\boldsymbol{x}) = - \begin{bmatrix} 
L_{Q_1}L_{\boldsymbol{f}}H_1 & L_{Q_2}L_{\boldsymbol{f}}H_1 & \cdots & L_{Q_6}L_{\boldsymbol{f}}H_1 \\ 
L_{Q_1}L_{\boldsymbol{f}}H_2 & L_{Q_2}L_{\boldsymbol{f}}H_2 & \cdots & L_{Q_6}L_{\boldsymbol{f}}H_2 
\end{bmatrix} \hfill
\end{gathered}
\end{equation*}
This configuration ensures that the effects of $\boldsymbol{w}$ are mitigated in the output channels during steady-state operation. The feedback gains $c_1^i$ and $c_2^i$ are determined by assigning the closed-loop poles of the $i$-th loop to $-\nu_{ci}$, where $\nu_{ci}$ denotes the designated controller bandwidth \cite{gao2003scaling}.

For clarity, the derivation of the first row of $\boldsymbol{\Gamma}(\boldsymbol{x})$ is illustrated using the first subsystem. By substituting the control law \eqref{eq: decoupling_control_law} into the system dynamics \eqref{eq: normal_form_subsystem}, the closed-loop representation for the first subsystem is obtained as:
\begin{equation}\label{eq: closed_loop_first_subsystem}
\begin{gathered}
\begin{bmatrix} \dot{\xi}_1 \\ \dot{\xi}_2 \end{bmatrix} = \begin{bmatrix} 0 & 1 \\ -c_1^1 & -c_2^1 \end{bmatrix} \begin{bmatrix} \xi_1 \\ \xi_2 \end{bmatrix} + \begin{bmatrix} 0 \\ 1 \end{bmatrix} c_1^1 Z_{\text{ref}} \hfill \\
+ \begin{bmatrix} \sum_{i=1}^{6} L_{Q_i}H_1 w_i \\ \sum_{i=1}^6 L_{Q_i}L_{\boldsymbol{f}}H_1 w_i \end{bmatrix} + \begin{bmatrix} 0 \\ \sum_{i=1}^6 \gamma_{1i}w_i \end{bmatrix} \hfill \\
y_1 = \begin{bmatrix} 1 & 0 \end{bmatrix} \begin{bmatrix} \xi_1 \\ \xi_2 \end{bmatrix} \hfill
\end{gathered}
\end{equation}

To ensure that the output $y_1$ remains unaffected by the disturbances during steady-state operation, the influence of the final two terms in \eqref{eq: closed_loop_first_subsystem} must be neutralized. This requirement is formulated as:
\begin{equation}\label{eq: steady_state_output_condition}
\begin{gathered}
\begin{bmatrix} 1 & 0 \end{bmatrix} \begin{bmatrix} 0 & 1 \\ -c_1^1 & -c_2^1 \end{bmatrix}^{-1} \begin{bmatrix} \sum_{i=1}^6 L_{Q_i}H_1 w_i \\ \sum_{i=1}^6 (L_{Q_i}L_{\boldsymbol{f}}H_1 + \gamma_{1i}) w_i \end{bmatrix} = 0 \hfill
\end{gathered}
\end{equation}

Solving \eqref{eq: steady_state_output_condition} for $\gamma_{1i}$ yields the specific decoupling gains required for each disturbance channel:
\begin{equation}\label{eq: gamma_gain_derivation}
\begin{gathered}
\gamma_{1i} = -L_{Q_i}L_{\boldsymbol{f}}H_1 - c_2^1 L_{Q_i}H_1, \quad \forall i=1, \dots, 6 \hfill
\end{gathered}
\end{equation}

\subsection{Inner-loop Attack-Resilient Tracking Problem Formulation}

Following the derivation of the decoupling gains in \eqref{eq: gamma_gain_derivation}, it is essential to consider the vulnerability of the baseline pseudo control component $\bm{u}_0$ to external cyberattacks. In this study, we investigate a scenario where the optimized pseudo-control input, intended to be the command $\bm{u}^*_0$, is targeted by an adversarial FDIA, denoted as $\boldsymbol{\delta}(t)$. Consequently, the actual control input that enters the system dynamics is the corrupted signal:
\begin{equation}\nonumber
    \bar{\bm{u}}^*_0 = \bm{u}^*_0 + \boldsymbol{\delta}
\end{equation}
where $\bm{u}^*_0$ is the resilient control law to be designed. While $\bar{\bm{u}}^*_0$ represents the corrupted signal effectively reaching the actuators, the proposed framework leverages its adaptive compensation capability to neutralize the adverse effects of $\boldsymbol{\delta}$. To ensure the analytical tractability of the system resilience, we establish the following condition regarding the attack signal:

\begin{assumption}
\label{ass: attack_structure}
The FDI attack signals denoted by $\delta_i(t)$ are exponentially unbounded. Their growth characteristics are constrained by a dynamic envelope defined as:
\[
\|\delta_i(t)\| \leqslant\exp(\bar\kappa_i t),
\]
where $\bar\kappa_i$ is a positive scalar constant.
\end{assumption}
The following definition facilitates the subsequent analysis. 

\begin{definition}[UUB\cite{khalil2002nonlinear}]
\label{def: UUB}
The signal $x(t)\in {\mathbb{R}^n}$ is said to be uniformly ultimately bounded (UUB) with the ultimate bound $b$, if there exist positive constants $b$ and $c$, independent of ${t_0} \geqslant 0$, and for every $a \in \left( {0,c} \right)$, there is $T = T\left( {a,b} \right) \geqslant 0$, independent of $t_0$, such that
\begin{equation}
\label{eq12}
\left\| {x\left( {{t_0}} \right)} \right\| \leqslant a\;\; \Rightarrow \;\;\left\| {x\left( t \right)} \right\| \leqslant b,\forall t \geqslant {t_0} + T
\end{equation}
\end{definition}

\begin{problem}[Inner-loop Attack-Resilient Tracking]
\label{pro: attack_resilient_tracking}
The objective is to design the optimized resilient pseudo-control vector $\boldsymbol{u}^*_0 = [u_{01}^*, u_{02}^*]^\top$ such that the tracking error states are regulated to a stable neighborhood of the equilibrium point in the sense of UUB stability.
\end{problem}

To mitigate the adverse effects of the structured adversarial signals defined in Assumption \ref{ass: attack_structure}, which can induce significant trajectory deviations or lead to catastrophic system failure if left uncompensated, we propose an optimization-based approach. The following section details the design of the RCLF-QP control framework. By integrating a resilient compensational term within a constrained QP solver, the framework dynamically counteracts malicious attack injections to maintain system stability and tracking performance.

\section{RCLF-QP Control Framework Design}\label{sec: resilient_qp}
Building upon the attack model and the constraints established in Assumption \ref{ass: attack_structure}, this section details the design of the RCLF-QP control framework. The primary objective of this framework is to design RCLF, where a resilient compensational term is developed to compensate for the maliciously injected unknown and unbounded FDIA with real-time online adaptation. By formulating RCLF as an affine constraint within a QP-based optimization problem, the framework provides a systematic and extensible resilient layer that effectively neutralizes the impact of adversarial FDIA injections.

The inner-loop tracking objective is addressed by decomposing the system into two autonomous tracking-error subsystems for altitude and reduced attitude regulation. We define the altitude error vector as $\bm{e}_z = [e_1, e_2]^\top = [Z - Z_{\text{ref}}, V_z]^\top$ and the reduced-attitude error vector as $\bm{e}_a = [e_3, e_4]^\top = [y_2, \dot{y}_2]^\top$. The resulting state-space representations are given by:
\begin{equation}\label{eq: error_subsystems}
    \dot{\bm{e}}_z = \underbrace{\begin{bmatrix} e_2 \\ 0 \end{bmatrix}}_{\bm{f}_{e,z}} + \underbrace{\begin{bmatrix} 0 \\ 1 \end{bmatrix}}_{\bm{g}_{e,z}} u_{01}^*, \quad
    \dot{\bm{e}}_a = \underbrace{\begin{bmatrix} e_4 \\ 0 \end{bmatrix}}_{\bm{f}_{e,a}} + \underbrace{\begin{bmatrix} 0 \\ 1 \end{bmatrix}}_{\bm{g}_{e,a}} u_{02}^*
\end{equation}
where $u_{01}^*$ and $u_{02}^*$ denote the optimized pseudo-control components determined by the RCLF-QP control framework.

To evaluate the stability of the tracking-error dynamics, we define two respective CLFs, $W_z$ and $W_a$, using a generalized quadratic form:
\begin{equation}\label{eq: generalized_clfs}
    W_z(\bm{e}_z) = \bm{e}_z^\top P_z \bm{e}_z, \qquad W_a(\bm{e}_a) = \bm{e}_a^\top P_a \bm{e}_a
\end{equation}
where $P_i \in \mathbb{R}^{2 \times 2}$ is a symmetric positive-definite matrix. 
The non-vanishing nature of $L_g W_i$ ensures that the optimized control inputs can directly regulate the energy dissipation rate along the stability boundary.

To account for adversarial influences caused by the malicious FDIA, an adaptive gain $\rho_i(t)$ is introduced for each channel. The real-time online evolution of this gain is governed by the following adaptation law (since the two subsystems in \eqref{eq: error_subsystems} are similar , for brevity, the following analysis is applied to both subsystems):
\begin{equation}\label{eq: adaptive_gain_law}
    \dot{\rho}_i(t) = \theta_i \|L_g W_i(\bm{e}_i)\|
\end{equation}
where $\theta_i > 0$ is the adaptation rate. Using these gains, we develop the following RCLF constraints:
\begin{equation}\label{eq: resilient_constraints_definition}
\begin{gathered}
    L_g W_z u_{01}^* \leqslant -C_z W_z - L_f W_z - \mathcal{G}_1(t) \triangleq b_1 \hfill \\
    L_g W_a u_{02}^* \leqslant -C_a W_a - L_f W_a - \mathcal{G}_2(t) \triangleq b_2 \hfill
\end{gathered}
\end{equation}
where the resilient compensational term $\mathcal{G}_i(t)$ is introduced to compensate for the adverse effects caused by the unknown and unbounded FDIA:
\begin{equation}\label{eq: compensation_term}
    \mathcal{G}_i(t) = \frac{(L_g W_i)(L_g W_i)^\top}{\|L_g W_i\| + e^{-\alpha t}} e^{ \rho_i(t)}
\end{equation}

The final optimized pseudo-control vector $\bm{u_0}^* = [u_{01}^*, u_{02}^*]^\top$ is obtained by solving a QP problem strictly adhering to the resilient stability constraint:
\begin{equation}\label{eq: clf_qp_optimization}
\begin{aligned}
    \min_{\bm{u_0}^* \in \mathbb{R}^2} \quad & \|\bm{u_0}^*-\bm{u}_0\|^2 \\
    \text{s.t.} \quad & \begin{bmatrix} L_g W_z & 0 \\ 0 & L_g W_a \end{bmatrix} \bm{u_0}^* \leqslant \begin{bmatrix} b_1 \\ b_2 \end{bmatrix}.
\end{aligned}
\end{equation}

\begin{theorem}
\label{thm: resilient_stability_guarantee}
Given Assumption \ref{ass: attack_structure}, and considering the quadrotor dynamical model \eqref{eq: nonlinear_system_dynamics} in the presence of lumped disturbances and adversarial attacks on the pseudo-control input channel $\boldsymbol{u}_0$, the inner-loop attack-resilient tracking problem is solved by the RCLF-QP \eqref{eq: clf_qp_optimization}. That is, Problem \ref{pro: attack_resilient_tracking} is solved.
\end{theorem}

\noindent\textbf{Proof of Theorem~\ref{thm: resilient_stability_guarantee}.}
To establish the stability of the inner-loop tracking under the RCLF-QP framework, we analyze the altitude error subsystem $\bm{e}_z$ and the reduced-attitude subsystem $\bm{e}_a$. Without loss of generality, we present the proof for the altitude subsystem using the Lyapunov function candidate $W_z = \bm{e}_z^\top P_z \bm{e}_z$ defined in \eqref{eq: generalized_clfs}. The analysis for the attitude subsystem follows an identical procedure.

The time derivative of $W_z$ along the trajectories of the closed-loop error dynamics \eqref{eq: error_subsystems} is given by:
\begin{equation}
\begin{gathered}
    \dot{W}_z = \nabla W_z \cdot \dot{\bm{e}}_z = L_f W_z + L_g W_z \bar{u}_{01}^*,
\end{gathered}
\end{equation}
where $\bar{u}_{01}^* = u_{01}^* + \delta_1(t)$ denotes the actual corrupted control input, with $\delta_1(t)$ representing the altitude component of the adversarial attack signal $\boldsymbol{\delta}(t)$. Substituting the corrupted input yields:
\begin{equation}
\begin{gathered}
    \dot{W}_z = L_f W_z + L_g W_z u_{01}^* + L_g W_z \delta_1.
\end{gathered}
\end{equation}
The RCLF-QP optimization in \eqref{eq: clf_qp_optimization} enforces the following RCLF constraint:
\begin{equation}
\begin{gathered}
    L_g W_z u_{01}^* \leqslant -C_z W_z - L_f W_z - \mathcal{G}_1(t).
\end{gathered}
\end{equation}
Substituting this constraint into the time derivative of $W_z$, we obtain:
\begin{equation}
\label{eq: Wz_dot_inequality}
\begin{gathered}
    \dot{W}_z \leqslant -C_z W_z + L_g W_z \delta_1 - \mathcal{G}_1(t).
\end{gathered}
\end{equation}
Substituting the explicit form of the resilient compensation term $\mathcal{G}_1(t)$ from \eqref{eq: compensation_term} into the inequality, we have:
\begin{equation}
\scalebox{0.9}{$
\label{eq: proof_manipulation_step1}
\begin{gathered}
    L_g W_z \delta_1 - \mathcal{G}_1(t) = L_g W_z \delta_1 - \frac{\|L_g W_z\|^2}{\|L_g W_z\| + e^{-\alpha t}} e^{\rho_1(t)} \\
    \leqslant \|L_g W_z\| \|\delta_1\| - \frac{\|L_g W_z\|^2}{\|L_g W_z\| + e^{-\alpha t}} e^{\rho_1(t)} \\
    = \frac{\|L_g W_z\| \left( \|L_g W_z\| \|\delta_1\| + e^{-\alpha t}\|\delta_1\| \right) - \|L_g W_z\|^2 e^{\rho_1(t)}}{\|L_g W_z\| + e^{-\alpha t}} \\
    = \frac{\|L_g W_z\| \left( \|L_g W_z\| \left( \|\delta_1\| - e^{\rho_1(t)} \right) + e^{-\alpha t}\|\delta_1\| \right)}{\|L_g W_z\| + e^{-\alpha t}}.
\end{gathered}
$}
\end{equation}
To ensure that $\dot{W}_z < 0$ outside a compact set, we require the term $L_g W_z \delta_1 - \mathcal{G}_1(t)$ to be non-positive. This condition is satisfied if the numerator in \eqref{eq: proof_manipulation_step1} is non-positive, which implies:
\begin{equation}
\begin{gathered}
    \|L_g W_z\| \left( e^{\rho_1(t)} - \|\delta_1\| \right) \geqslant e^{-\alpha t} \|\delta_1\|.
\end{gathered}
\end{equation}
The adaptive gain $\rho_1(t)$ evolves according to $\dot{\rho}_1 = \theta_1 \|L_g W_z\|$. As long as tracking errors persist (i.e., $\|L_g W_z\| \neq 0$), $\rho_1(t)$ increases. Given the exponential nature of the compensation term $e^{\rho_1(t)}$ and considering Assumption~\ref{ass: attack_structure}, it will eventually dominate the growth of $\|\delta_1(t)\|$. Furthermore, the term $e^{-\alpha t}\|\delta_1(t)\|$ decays to zero as $t \to \infty$. Consequently, there exists a time $T^* > 0$ such that for all $t > T^*$, that is when $\theta_1 \|L_g W_z\|>\bar\kappa_i$, the condition $e^{\rho_1(t)} > \|\delta_1(t)\|$ holds.   . We define the compact sets $\Upsilon_z \equiv \{ \bm{e}_z : \|L_g W_z\| \leqslant \epsilon_1 \}$, where $\epsilon_1$ represent small positive constants related to the residual errors during the convergence phase.

Outside the set $\Upsilon_z$, and for $t > T^*$, the term $L_g W_z \delta_1 - \mathcal{G}_1(t)$ becomes negative, yielding:
\begin{equation}
\begin{gathered}
    \dot{W}_z \leqslant -C_z W_z < 0.
\end{gathered}
\end{equation}
Consequently, $\dot{W}_z$ is strictly negative-definite outside the compact set $\Upsilon_z$ for all $t > T^*$. By invoking LaSalle's invariance principle \cite{krstic1995nonlinear}, we conclude that the error trajectories $\bm{e}_z$ (and similarly $\bm{e}_a$) are UUB. Thus, the proposed RCLF-QP framework guarantees that the system states are regulated to a stable neighborhood of the equilibrium despite the presence of unbounded FDIA. This completes the proof. $\hfill\blacksquare$    


\section{Simulation Results}\label{sec: simulation}

The effectiveness of the proposed resilient control architecture is evaluated through numerical simulations. To highlight the robustness provided by the integrated MB-ESO and RCLF-QP control framework, we compare our methodology with the baseline pseudo control law $u_0$ described in \cite{sun2021incremental, yang2012nonlinear}. This comparison specifically examines the system's ability to maintain stability when subjected to malicious signal injections and lumped environmental disturbances.
The simulation is conducted using the physical parameters of the quadrotor detailed in Table \ref{tab: parameters}. To test the resilience of the framework, an adversarial attack signal $\boldsymbol{\delta}$ is injected into the pseudo-control input channel. This attack is designed as an unbounded signal intended to compromise the inner-loop stability.

\begin{table}[htbp]
    \caption{Physical Parameters of the Quadrotor}
    \label{tab: parameters}
    \centering
    \small
    \begin{tabular}{@{}cc@{\hspace{4mm}}cc@{}}
        \toprule
        Parameter & Value & Parameter & Value\\
        \midrule
        \(I_x\) & \(1.4 \times 10^{-3}\) kg\(\cdot\)m\textsuperscript{2} & \(m_v\) & \(0.375\) kg \\
        \(I_y\) & \(1.3 \times 10^{-3}\) kg\(\cdot\)m\textsuperscript{2} & \(b\)   & \(0.115\) m \\
        \(I_z\) & \(2.5 \times 10^{-3}\) kg\(\cdot\)m\textsuperscript{2} & \(l\)   & \(0.0875\) m \\
        \(\kappa_0\) & \(1.9 \times 10^{-6}\) kg\(\cdot\)m/rad\textsuperscript{2} & \(g\) & \(9.81\) m/s\textsuperscript{2}\\
        \(\tau_0\) & \(1.9 \times 10^{-8}\) kg\(\cdot\)m\textsuperscript{2}/rad\textsuperscript{2} & \(\gamma\) & \(1.92 \times 10^{-3}\) N\(\cdot\)m\(\cdot\)s \\
        \(\chi\) & 105\(^\circ\) &   &   \\
        \bottomrule
    \end{tabular}
\end{table}

The operational constraints include a maximum tilt angle of 60 degrees and rotor speed bounds between 0 rad/s and 2000 rad/s. Sensor measurements incorporate Band-Limited White Noise with power values of \(1\times 10^{-8}\), \(1\times 10^{-8}\), \(2\times 10^{-4}\), and \(1\times 10^{-5}\) for position, attitude, angular velocity, and rotor speed, respectively. The specific bandwidths and gains are adopted from Table II in \cite{chen2024quadrotor}, whereas the baseline pseudo controller $u_0$ utilizes identical gains but operates without the RCLF-QP mechanism.
For the simulation scenario presented here, a ``gate" reference trajectory is considered for the horizontal position, where the desired position transitions from zero to a constant setpoint of $1$\,m at $t=4\,\mathrm{s}$ and subsequently returns to zero at $t=8\,\mathrm{s}$. The altitude reference is held constant at $Z_{ref} = 0\,\mathrm{m}$. Starting at $t=10\,\mathrm{s}$, an unbounded FDIA is injected. The attack signal applied to rotor~1 is $\delta_1(t)=0.5t$, and the one to the opposing rotor is $\delta_2(t)=0.75t$.
\subsection{Trajectory Tracking and Reduced Attitude}

The tracking performance under this adversarial attack is illustrated in Fig.~\ref{fig: trajectory_comparison}. As shown in the first three subplots of Fig.~\ref{fig: trajectory_comparison}, the quadrotor successfully tracks the gate reference in the $X$ and $Y$ axes and maintains the desired altitude $Z$ prior to the attack. When subjected to the unbounded FDIAs at $t=10\,\mathrm{s}$, the baseline pseudo controller $u_0$ (blue dotted line) fails to maintain tracking performance. The lack of resilience-enforcing constraints allows the adversarial input to dominate the closed-loop error dynamics, resulting in significant divergence. In contrast, the proposed RCLF framework (red solid line) preserves stable and accurate tracking by incorporating the resilient control Lyapunov function within the QP-based optimization.

The bottom subplot of Fig.~\ref{fig: trajectory_comparison} displays the reduced attitude output $y_2$. In the reduced attitude control framework adopted for failure conditions, driving $y_2 \to 0$ is critical for ensuring the internal zero dynamics remain stable. The proposed RCLF framework successfully maintains $y_2$ within a narrow bounded error band near zero, even after the attack is triggered. While $y_2$ diverges using the baseline controller (blue line).


The resulting spatial path is visualized in Fig.~\ref{fig: 3d_trajectory}. The 3D plot illustrates the resilience of the proposed controller under FDIA. While the baseline trajectory exhibits unbounded divergence leading to a total system crash, the RCLF-QP framework maintains stability. Although a residual tracking error is observable during the attack, the controller effectively prevents the catastrophic drift seen in the baseline case, maintaining a stable hover near an altitude of $0.2\,\mathrm{m}$.

\begin{figure}[!htbp]
    \centering
    \includegraphics[width=0.95\linewidth]{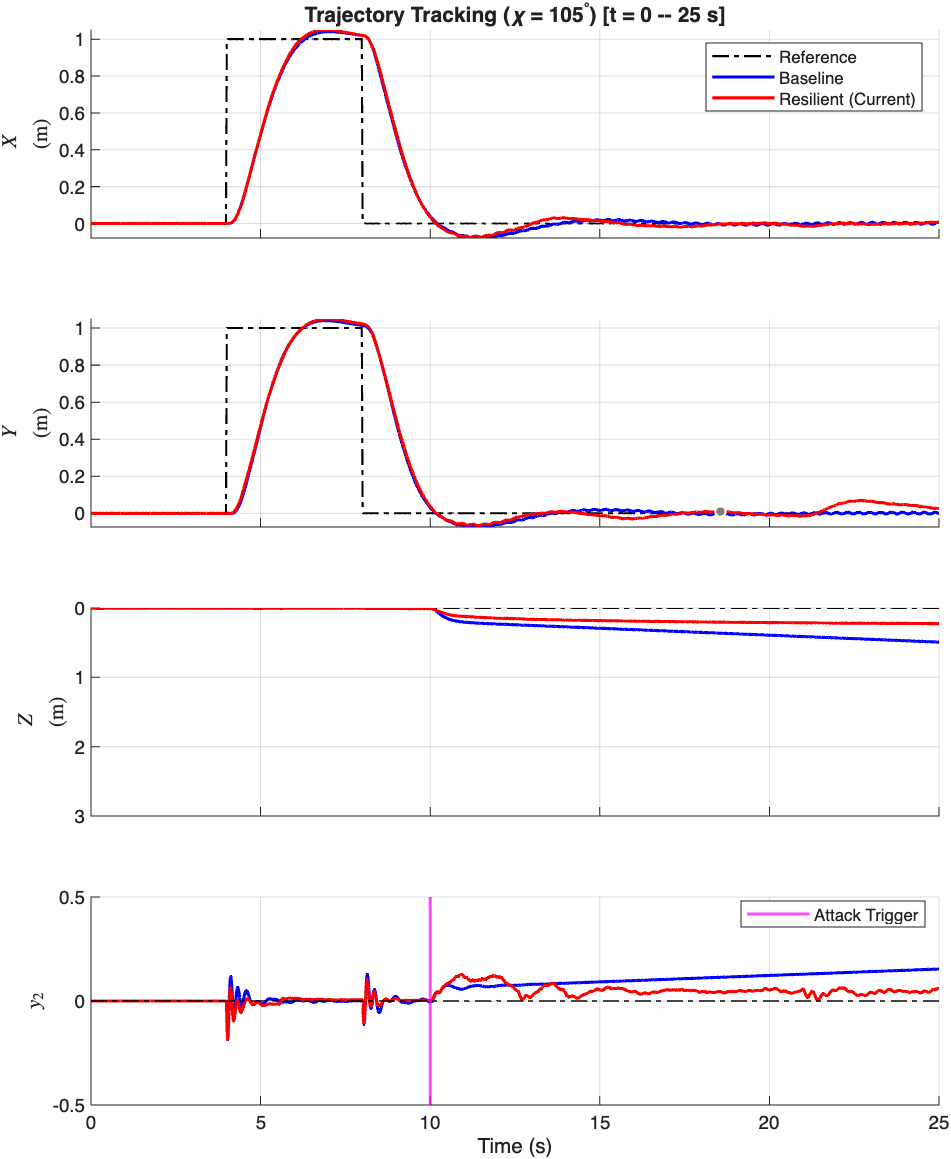}
    \caption{Trajectory tracking comparison for $X$, $Y$, $Z$ positions and the reduced attitude output $y_2$. The unbounded attack is triggered at $t=10\,\mathrm{s}$, where the attack signals applied to rotor~1 and the opposing rotor~3 are $\delta_1(t)=0.5t$ and $\delta_2(t)=0.75t$, respectively.}
    \label{fig: trajectory_comparison}
\end{figure}

\begin{figure}[!htbp]
    \centering
    \includegraphics[width=0.95\linewidth]{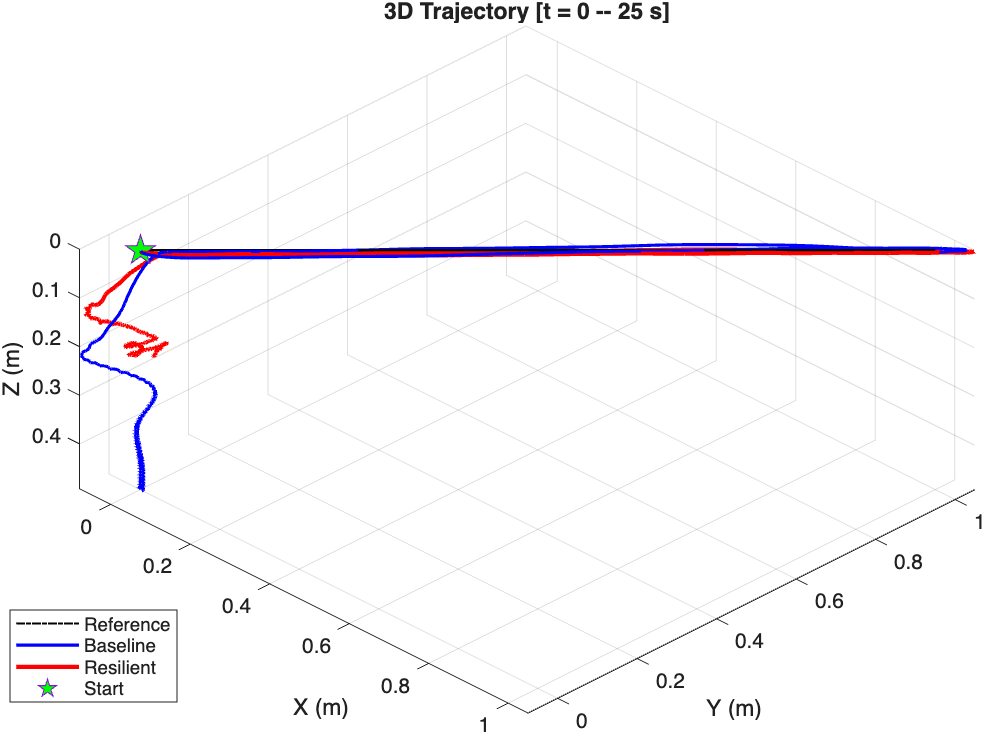}
    \caption{3D flight trajectories under adversarial attack. Comparison between the baseline controller and the resilient RCLF-QP controller.}
    \label{fig: 3d_trajectory}
\end{figure}

\section{Conclusion}
This letter has presented a resilient control architecture to secure quadrotor's operations against three types of threats: actuator failures, lumped external disturbances caused by aerodynamics and strong wind, and malicious cyberattacks. A RCLF-QP framework has been developed to mitigate the adverse effects caused by unknown and unbounded cyberattacks online in real-time. Unlike conventional PD feedback control, such as in \cite{chen2024quadrotor}, the proposed optimization-based approach offers a systematic and extensible framework. By formulating the QP problem with resilient components as hard constraints against unknown and unbounded attack signals, the architecture maintains rigorous stability guarantees while allowing for seamless integration of new control objectives. High-fidelity simulations demonstrate that the RCLF-QP mechanism effectively prevents trajectory divergence and system failure in scenarios where baseline methods succumb to attacks. Ultimately, this work provides a foundation for enhancing the autonomy and security of unmanned aerial systems in increasingly adversarial environments.



 


\bibliographystyle{asmejour}   

\bibliography{asmejour-sample} 



\end{document}